\documentclass[11pt]{article}

\usepackage[margin=1in]{geometry}
\usepackage{amsmath,amssymb,amsthm,mathtools}
\usepackage{booktabs}
\usepackage{enumitem}
\usepackage[numbers,sort&compress]{natbib}
\usepackage[colorlinks=true,linkcolor=blue,citecolor=blue,urlcolor=blue]{hyperref}
\usepackage{microtype}

\theoremstyle{plain}
\newtheorem{theorem}{Theorem}

\newtheorem{corollary}[theorem]{Corollary}

\theoremstyle{definition}
\newtheorem{definition}[theorem]{Definition}
\newtheorem{assumption}[theorem]{Assumption}

\theoremstyle{remark}
\newtheorem{remark}[theorem]{Remark}

\newcommand{\R}{\mathbb{R}}
\newcommand{\E}{\mathbb{E}}
\newcommand{\Pbb}{\mathbb{P}}
\newcommand{\F}{\mathcal{F}}
\newcommand{\Hh}{\mathcal{H}}
\newcommand{\A}{\mathcal{A}}
\newcommand{\Bset}{\mathcal{B}}
\newcommand{\M}{\mathcal{M}}

\newcommand{\doop}{\operatorname{do}}
\newcommand{\ES}{\operatorname{ES}}

\title{\vspace{-1.25cm}
Foundations of a Time-Consistent Counterfactual\\
Actuarial Runtime for Autonomous AI Agents\\
\large Foundational working paper}
\author{Hao-Hsuan Chen\\
\small Department of Risk Management and Insurance\\
\small National Chengchi University\\
\small \texttt{dannychen9199@gmail.com}}
\date{May 2026}

\begin{document}
\maketitle

\begin{abstract}
We propose a foundational runtime actuarial layer for autonomous AI agents in
which every side-effect-bearing action carries a time-consistent,
counterfactual risk toll computed against a contractually fixed safe default,
inside an explicit underwriting boundary. The framework treats per-action
insurance as the primary unit of analysis and replaces post-hoc annual
liability cover with a pre-action transaction layer. The paper establishes
four structural results: (i) a well-defined counterfactual toll under a
chosen safe-default mapping and continuation policy, with explicit
non-uniqueness; (ii) a no-splitting property within an underwriting boundary
that telescopes path-decomposed actions into a boundary potential, with a
corollary tying gaming-resistance to boundary design; (iii) an
irreversible-authority premium, split into a strictly positive action-level
component and an if-and-only-if characterisation of the set-level robust
capital increase; and (iv) a conservative runtime gating theorem that
translates high-probability toll envelopes into an executed-action budget
guarantee. The result is the mathematical base layer for a broader program:
an empirical companion instantiates the runtime through an Actuarial Action
Interface and authority-frontier experiments; a mechanism-design companion
studies strategic operator incentives and cross-boundary aggregation; and a
dynamic-underwriting companion studies experience rating and audit-replay
calibration. The present paper states the primitive contract, the toll
identity, the within-boundary no-arbitrage result, and the budget guarantee
on which those later layers depend.
\end{abstract}

\section{Positioning}\label{sec:positioning}

The research question is:
\begin{quote}
\itshape How can insurance, warranty, or escrow-backed guarantees be embedded
directly into an autonomous AI agent runtime, so that each economically
meaningful action is priced, budgeted, and gated before execution?
\end{quote}

The intended contribution is not another audit score for AI systems. The
object is a \emph{transaction layer}: every side-effect-bearing action pays
a risk toll, or it is blocked, downgraded to a safe default, or escalated to
human approval.

This paper deliberately avoids three common traps.
\begin{enumerate}[leftmargin=1.5em,itemsep=2pt]
  \item \textbf{Static tail-risk allocation is not new.} Euler allocation,
  marginal expected shortfall, and expected-shortfall capital attribution are
  standard tools \citep{tasche2007euler,acharya2017measuring,mcneil2015qrm}.
  The novelty here is the \emph{dynamic, action-conditioned, runtime} object,
  not the use of a coherent risk measure.
  \item \textbf{Conditional CVaR is not automatically time-consistent.}
  Runtime decisions require a dynamic risk measure compatible with Bellman's
  principle, not a naive sequence of one-period CVaR calculations
  \citep{roorda2007time,bion2008dynamic,ruszczynski2010risk}.
  \item \textbf{Agent actions are not tokens.} The unit of pricing is a
  side-effect-bearing tool call with external economic, legal, or operational
  consequences, not internal chain-of-thought, hidden deliberation, or token
  generation.
\end{enumerate}

\paragraph{Role in the research program.}
This paper is the foundation layer. It defines the priced object
(a side-effect-bearing action), the counterfactual comparison (execute the
action versus execute the safe default), the aggregation unit (an
underwriting boundary), and the pathwise budget guarantee. Later layers use
these objects without changing them: the empirical runtime layer
\citep{chen2026insuring} implements the gate and measures authority release;
the mechanism-design layer asks
which contract clauses make strategic manipulation unprofitable; and the
dynamic-underwriting layer asks how premiums and reserves should update as
experience accumulates. This separation is deliberate: Paper~A proves the
primitive actuarial accounting identities, while companion work tests,
extends, and strategically hardens them.

\section{Related Work}\label{sec:related}

\paragraph{Algorithmic and AI insurance.}
\citet{bertsimas2021algorithmic} price liability exposure for machine-driven
classification decisions and connect model properties to insurance value;
\citet{bertsimas2024catastrophe} extends adaptive robust optimisation to
catastrophe pricing. Recent work on trustworthy AI agents proposes
transaction-level compensation and underwriting-inspired standards for
agent-mediated transactions \citep{hua2026quantifying}. This paper differs
by making the priced object an individual side-effect action inside a
runtime control loop, rather than an annual policy or product warranty.

\paragraph{Dynamic risk measures.}
Time consistency is central. Conditional CVaR and VaR can fail dynamic
consistency in multistage settings \citep{roorda2007time,kang2006time}.
Dynamic coherent or convex risk measures restore a version of Bellman's
principle through recursive conditional risk mappings
\citep{bion2008dynamic,cheridito2006dynamic,detlefsen2005conditional,
ruszczynski2010risk,shapiro2009lectures}. Risk-sensitive Markov decision
processes provide the algorithmic counterpart
\citep{chow2014algorithms,tamar2015policy}.

\paragraph{Causal credit assignment.}
The proposed toll is counterfactual: it compares the loss distribution
under executing an action with the loss distribution under a safe default
using an interventional model \citep{pearl2009causality}. In deployed
systems, the interventional distribution can be approximated by simulator
rollouts, causal world models, or off-policy estimators.

\paragraph{Runtime guarantees.}
The implementation goal is not exact real-time tail-risk computation for
every action. It is conservative gating: fast upper bounds for routine
actions, and deeper simulation or human escalation for high-stakes actions.
Conformal prediction can provide finite-sample high-probability upper
envelopes under specified calibration assumptions
\citep{vovk2005algorithmic,angelopoulos2021gentle}; online adaptive
selection requires adaptive variants \citep{gibbs2021adaptive}.

\section{Model}\label{sec:model}

\subsection{Histories and side-effect actions}

Fix a finite horizon $T$. Let $(\Omega, \F, (\F_t)_{t=0}^T, \Pbb)$ be a
filtered probability space with $\F_0$ trivial. At time $t$ the agent has
history $h_t \in \Hh_t$ represented by $\F_t$, and selects
$a_t \in \A_t(h_t)$.

\begin{definition}[Side-effect-bearing action]\label{def:side-effect}
Decompose the environment state at time $t$ as
$s_t = (s_t^{\mathrm{int}}, s_t^{\mathrm{ext}})$, where $s_t^{\mathrm{ext}}$
collects contractually designated external state components such as
financial ledgers, persistent databases, third-party communications,
legal commitments, security permissions, and production systems. An
action $a \in \A_t(h_t)$ is \emph{side-effect-bearing} relative to an
underwriting boundary if the interventional transition kernel of
$s_{t+1}^{\mathrm{ext}}$ under $\doop(a)$ differs from the kernel under
the null (or contractually identified no-op) action on at least one
designated external component, and the induced change is not freely and
instantaneously reversible by the agent alone within contractually
specified time and cost thresholds. Examples: payment, refund, order
placement, database write, email send, form submission, code deployment,
trade execution, claim adjudication. Internal reasoning, hidden
deliberation, and token generation are not priced actions.
\end{definition}

\begin{definition}[Minimal-authority safe default]\label{def:safe-default}
For an intended action $a \in \A_t(h_t)$, the safe default $a^0_t(h_t)$ is
the ex ante fixed, least externally committing action in $\A_t(h_t)$ that
preserves the same intent while minimising external authority. The mapping
$a \mapsto a^0_t(h_t)$ is part of the contract and \emph{cannot be
selected by the operator after observing the toll}; this prevents the
adversary from picking the safe default that minimises its own
counterfactual gap. Examples: ``send email'' $\mapsto$ ``save draft'';
``execute trade'' $\mapsto$ ``paper trade''; ``pay invoice'' $\mapsto$
``draft payment for approval''; ``drop table'' $\mapsto$ ``log proposed
change without execution''.
\end{definition}

\subsection{Counterfactual losses}

Let $L \in L^p_+(\F_T)$, $p \in [1, \infty]$, be terminal loss. For an
intervention $\doop(a)$ at history $h_t$, let $L^{\doop(a), \pi^+}$ denote
terminal loss when action $a$ is forced at time $t$ and a fixed continuation
policy $\pi^+$ is followed thereafter. We make the continuation policy
explicit in notation throughout. The interventional distribution may be
supplied by:
\begin{enumerate}[leftmargin=1.5em,itemsep=2pt]
  \item a sandbox simulator (reference implementation in companion Paper~B);
  \item a causal world model $s_{t+1} = f(s_t, a_t, \xi_t)$;
  \item an off-policy estimator from logged trajectories.
\end{enumerate}

\begin{assumption}[Interventional well-posedness]\label{ass:wellposed}
For each priced action $a \in \A_t(h_t)$ and safe default $a^0_t(h_t)$, the
interventional terminal losses $L^{\doop(a), \pi^+}$ and
$L^{\doop(a^0_t(h_t)), \pi^+}$ are well-defined elements of $L^p_+(\F_T)$
under an admissible family of environment models $\M$.
\end{assumption}

\subsection{Time-consistent dynamic risk}

Let $\sigma_t: L^p(\F_{t+1}) \to L^p(\F_t)$ be a one-step conditional risk
mapping. The dynamic risk process is defined recursively:
\begin{equation}\label{eq:recursive-risk}
\rho_T(X) = X, \qquad
\rho_t(X) = \sigma_t\bigl(\rho_{t+1}(X)\bigr), \quad t = T-1, \ldots, 0.
\end{equation}

\begin{assumption}[Convex, time-consistent runtime risk]\label{ass:rho}
Each one-step mapping $\sigma_t$ satisfies, for all $X, Y \in L^p(\F_{t+1})$
and $\F_t$-measurable bounded $m$:
\begin{enumerate}[label=(\roman*),leftmargin=1.5em,itemsep=2pt]
  \item (normalisation) $\sigma_t(0) = 0$;
  \item (monotonicity) $X \le Y \Rightarrow \sigma_t(X) \le \sigma_t(Y)$;
  \item (locality) $\sigma_t(\mathbf{1}_A X) = \mathbf{1}_A \sigma_t(X)$ for
  $A \in \F_t$;
  \item (translation invariance) $\sigma_t(X + m) = \sigma_t(X) + m$;
  \item (convexity) $\sigma_t(\lambda X + (1-\lambda) Y) \le \lambda \sigma_t(X)
  + (1-\lambda)\sigma_t(Y)$, $\lambda \in [0,1]$.
\end{enumerate}
\emph{Coherent specialisations} additionally impose positive homogeneity
$\sigma_t(\lambda X) = \lambda \sigma_t(X)$ for $\lambda \ge 0$
\citep{artzner1999coherent}. The entropic special case
$\sigma_t(Y) = \gamma^{-1}\log\E[\exp(\gamma Y)\mid \F_t]$, $\gamma > 0$,
satisfies (i)--(v) and yields a time-consistent recursive composition,
but is \emph{not coherent} because positive homogeneity fails in general.
We use the entropic family as the reference runtime mapping throughout;
results below require only the convex case unless explicitly stated.
\end{assumption}

\subsection{Underwriting boundary}

\begin{definition}[Underwriting boundary]\label{def:boundary}
An underwriting boundary $B \in \Bset$ is a contractually defined unit of
aggregation, e.g.
\[
B = (\text{legal entity},\; \text{rolling time window},\;
     \text{action category},\; \text{currency or asset class}).
\]
Let $E^B_t \in \R^d_+$ be the observable cumulative exposure state inside
$B$ at time $t$, and let $\Delta E^B_t(a_t) \in \R^d_+$ be the exposure
increment contributed by action $a_t$.
\end{definition}

\begin{assumption}[Boundary observability]\label{ass:boundary-obs}
For each boundary $B$, the cumulative exposure $E^B_t$ and increment
$\Delta E^B_t(a_t)$ are observable, contractually attributable, and updated
before the next action priced in $B$ is evaluated.
\end{assumption}

\section{Counterfactual Action Toll}\label{sec:toll}

\begin{definition}[Counterfactual action toll]\label{def:toll}
Given history $h_t$, continuation policy $\pi^+$, dynamic risk process
$\rho_t$, and safe-default mapping $a^0$, the action-level counterfactual
risk toll is
\begin{equation}\label{eq:toll}
c_t(a \mid h_t; \pi^+, \rho, a^0)
\;=\;
\rho_t\!\bigl(L^{\doop(a), \pi^+}\bigr)
\;-\;
\rho_t\!\bigl(L^{\doop(a^0_t(h_t)), \pi^+}\bigr).
\end{equation}
We write $c_t^+(a) := \max\{c_t(a), 0\}$ for the runtime charge under the
convention that risk-reducing actions are not subsidised.
\end{definition}

\begin{theorem}[Time-consistent counterfactual tolls]\label{thm:tc-tolls}
Under Assumptions~\ref{ass:wellposed}--\ref{ass:rho} and the recursion
\eqref{eq:recursive-risk}, the dynamic risk process is time-consistent:
for every $0 \le t < s \le T$, if $\rho_s(X) \le \rho_s(Y)$ a.s.\ then
$\rho_t(X) \le \rho_t(Y)$ a.s. Moreover, for every well-posed action
$a \in \A_t(h_t)$, the toll $c_t(a \mid h_t; \pi^+)$ in \eqref{eq:toll}
is $\F_t$-measurable and bounded. If
$L^{\doop(a), \pi^+} \le L^{\doop(b), \pi^+}$ a.s.\ for two actions $a, b$,
then $\rho_t(L^{\doop(a), \pi^+}) \le \rho_t(L^{\doop(b), \pi^+})$.
\end{theorem}

\begin{proof}[Proof sketch]
By backward induction. Each $\sigma_s$ maps $L^p(\F_{s+1}) \to L^p(\F_s)$,
so starting from $\rho_T(X) = X \in L^p(\F_T)$ and applying $\sigma_{T-1},
\ldots, \sigma_t$, we obtain $\rho_t(X) \in L^p(\F_t)$; in particular,
$\rho_t(X)$ is $\F_t$-measurable. For dynamic consistency, suppose
$\rho_{s}(X) \le \rho_{s}(Y)$ a.s.; then monotonicity of $\sigma_{s-1}$
(Assumption~\ref{ass:rho}(ii)) gives
$\rho_{s-1}(X) = \sigma_{s-1}(\rho_s(X)) \le \sigma_{s-1}(\rho_s(Y))
 = \rho_{s-1}(Y)$ a.s. Iterating from $s$ down to $t$ yields the claim.
The toll $c_t$ is the $\F_t$-measurable difference of two $L^p(\F_t)$
random variables and is bounded by translation invariance:
$|c_t(a)| \le \|L^{\doop(a), \pi^+}\|_\infty +
\|L^{\doop(a^0), \pi^+}\|_\infty$ when $p = \infty$.
The dominance claim follows from dynamic consistency with $s = T$ and
$\rho_T(X) = X$. No uniqueness of the toll is used; it is well-defined
only after fixing $(a^0, \pi^+, \rho, \Bset, \M)$.
\end{proof}

\begin{remark}[No uniqueness claim]\label{rem:no-unique}
Theorem~\ref{thm:tc-tolls} does not claim that $\{c_t\}$ is the unique
actuarial decomposition of $\rho_0(L)$ satisfying any prescribed axiom set;
competing decompositions (e.g.\ Shapley-style or martingale-based) are
admissible. The toll is a \emph{well-defined}, time-consistent,
counterfactual decomposition once the five primitives are fixed.
\end{remark}

\section{No-Splitting within an Underwriting Boundary}\label{sec:no-splitting}

Let $\Phi_B: \R^d_+ \to \R_+$ be a monotone toll potential with
$\Phi_B(0) = 0$. Define the boundary-level toll
\begin{equation}\label{eq:potential-toll}
\lambda_t\bigl(a_t \mid E^B_t, B\bigr)
\;=\;
\Phi_B\bigl(E^B_t + \Delta E^B_t(a_t)\bigr) - \Phi_B\bigl(E^B_t\bigr).
\end{equation}
The boundary toll is the operational realisation of $c_t$ when the
conditional loss law given $\F_t$ depends on within-$B$ history only
through $E^B_t$; the calibration of $\Phi_B$ to $\rho_t$ is treated in
Paper~B.

\begin{assumption}[Boundary-determined conditional loss]\label{ass:Bdet}
Inside boundary $B$, the conditional law of $L$ given $\F_t$ depends on
the within-$B$ action history only through the cumulative exposure
$E^B_t$ and an outside state variable $\xi_t$:
$\Pbb(L \in \cdot \mid \F_t) = \Pbb(L \in \cdot \mid E^B_t, \xi_t)$.
The variable $\xi_t$ is $\F_t$-measurable, encodes state outside~$B$ that
is relevant to the priced conditional law, and is \emph{not reset by
splitting an action into smaller increments, opening new sessions,
re-labelling equivalent actions, or routing through proxy agents}. The
non-reset property is part of the contract definition of $\xi_t$; it
fails by default for naive boundary designs and is the principal
obligation of boundary calibration.
\end{assumption}

\begin{theorem}[No-splitting within an underwriting boundary]\label{thm:no-split}
Fix a boundary $B$ satisfying Assumptions~\ref{ass:boundary-obs}--\ref{ass:Bdet}.
For any finite sequence of within-$B$ side-effect actions $(a_0, \ldots,
a_{T-1})$ with exposure dynamics $E^B_{t+1} = E^B_t + \Delta E^B_t(a_t)$,
the total boundary toll generated by \eqref{eq:potential-toll} satisfies
\[
\sum_{t=0}^{T-1} \lambda_t(a_t \mid E^B_t, B)
\;=\;
\Phi_B(E^B_T) - \Phi_B(E^B_0).
\]
In particular, any decomposition of the same cumulative exposure increment
$E^B_T - E^B_0$ into smaller actions within the same boundary pays the
same total toll. If $\Phi_B$ is convex and nondecreasing, concentrating
exposure into higher cumulative states cannot reduce total toll relative
to the terminal cumulative-exposure charge.
\end{theorem}

\begin{proof}[Proof sketch]
Substitute the dynamics $E^B_{t+1} = E^B_t + \Delta E^B_t(a_t)$ into
\eqref{eq:potential-toll} to get
$\lambda_t = \Phi_B(E^B_{t+1}) - \Phi_B(E^B_t)$, then telescope:
$\sum_{t=0}^{T-1}[\Phi_B(E^B_{t+1}) - \Phi_B(E^B_t)] = \Phi_B(E^B_T) -
\Phi_B(E^B_0)$. The total depends on $(E^B_0, E^B_T)$ alone, hence is
invariant under within-$B$ decomposition. Convexity of $\Phi_B$ is not
required for the identity; it becomes material when the designer wants
the marginal toll to increase with cumulative exposure. If
Assumption~\ref{ass:Bdet} fails, $L$ admits a path-dependent component
beyond $E^B_T$ which $\Phi_B$ cannot capture, and adversarial sequencing
can strictly reduce realised toll within~$B$. Cross-boundary arbitrage is
not blocked by Theorem~\ref{thm:no-split} and must be addressed by
boundary design (Corollary~\ref{cor:gaming}).
\end{proof}

\begin{corollary}[Boundary design as gaming-resistance]\label{cor:gaming}
Theorem~\ref{thm:no-split} prevents splitting only within a fixed
observable boundary whose state $(E^B_t, \xi_t)$ is sufficient for the
relevant conditional loss law. Cross-session, cross-entity,
path-dependent, $\xi_t$-relabel, and proxy-agent arbitrage are not
automatically blocked; they must be handled by boundary design, including
common-control aggregation, rolling time windows, related-party rules,
action-category linking, currency and asset normalisation, and audit
rules that prevent the operator from resetting $\xi_t$ by relabelling
economically equivalent actions. If Assumption~\ref{ass:Bdet} fails
inside~$B$, gaming-resistance reduces to redesigning $B$ so the
assumption holds for the contractually insured loss class.
\end{corollary}

\section{Irreversible Authority Premium}\label{sec:iap}

We compare two authority sets $\A^0_t(h_t) \subset \A^+_t(h_t)
= \A^0_t(h_t) \cup \{a^+\}$. We work under an ambiguity set $\M$ of
admissible environment models, with $\rho_t^M$ denoting the dynamic risk
under model $M \in \M$.

\begin{definition}[Uncompensated irreversible tail exposure]\label{def:irr}
The action $a^+$ creates uncompensated irreversible tail exposure relative
to safe default $a^0$ at history $h_t$ if there exist a model
$M^\star \in \M$, an event $G \in \F_T$ with $\Pbb^{M^\star}(G \mid h_t)
> 0$, and constants $\delta > 0$, $\eta < \delta$ such that:
\begin{enumerate}[label=(\roman*),leftmargin=1.5em,itemsep=2pt]
  \item \emph{Tail-loss gap.}
  $L^{\doop(a^+), \pi^+} \ge L^{\doop(a^0), \pi^+}$ a.s.\ under $M^\star$,
  and $L^{\doop(a^+), \pi^+} \ge L^{\doop(a^0), \pi^+} + \delta$ on $G$
  under $M^\star$.
  \item \emph{Irreversibility.} No admissible continuation policy can
  reduce the loss gap on $G$ by more than $\eta < \delta$.
  \item \emph{Strict $\rho$-monotonicity for this pair under $M^\star$.}
  A terminal-loss change that is weakly nonnegative a.s.\ and strictly
  positive on $G$ produces a strictly higher value of $\rho_t^{M^\star}$
  at the current history.
\end{enumerate}
The three conditions are logically independent: (i)~alone leaves room for
the gap to be erased by continuation, (i)+(ii) without (iii) leaves the
gap invisible to the risk mapping, and (i)+(iii) without (ii) leaves the
gap hedgeable.
\end{definition}

Define the action-level authority premium and the set-level robust capital:
\begin{align}
\Delta_t(a^+ \mid h_t)
&\;:=\;
\sup_{M \in \M}\bigl[\rho_t^M(L^{\doop(a^+), \pi^+})
- \rho_t^M(L^{\doop(a^0), \pi^+})\bigr]_+ ,
\label{eq:Delta} \\[2pt]
K_t(U \mid h_t)
&\;:=\;
\sup_{M \in \M}\sup_{a \in U}
\rho_t^M(L^{\doop(a), \pi^+}),
\quad U \subseteq \A_t(h_t).
\label{eq:K}
\end{align}

\begin{theorem}[Irreversible authority premium]\label{thm:iap}
Suppose $\A^+_t(h_t) = \A^0_t(h_t) \cup \{a^+\}$, the relevant robust
risks are finite (so that $K_t(\A^+_t \mid h_t) < \infty$ and the
supremum defining $\Delta_t(a^+ \mid h_t)$ is finite), and $a^+$ creates
uncompensated irreversible tail exposure relative to its safe default
$a^0$ in the sense of Definition~\ref{def:irr}. Then:
\begin{enumerate}[label=\textup{(\alph*)},leftmargin=1.8em,itemsep=2pt]
  \item \emph{Action-level premium is strictly positive:}
  $\Delta_t(a^+ \mid h_t) > 0$.
  \item \emph{Set-level capital increase is conditional:}
  $K_t(\A^+_t(h_t) \mid h_t) > K_t(\A^0_t(h_t) \mid h_t)$ if and only if
  \[
  \sup_{M \in \M} \rho_t^M\bigl(L^{\doop(a^+), \pi^+}\bigr)
  \;>\;
  K_t(\A^0_t(h_t) \mid h_t).
  \]
\end{enumerate}
\end{theorem}

\begin{proof}[Proof sketch]
\emph{Action-level (a).} Pick the witness $M^\star$ of
Definition~\ref{def:irr}. Set $X = L^{\doop(a^+), \pi^+}$ and
$Y = L^{\doop(a^0), \pi^+}$. Condition~(i) gives $X \ge Y$ a.s.\ under
$M^\star$ and $X \ge Y + \delta$ on $G$. Condition~(ii) ensures the gap on
$G$ cannot be erased by continuation policy, so even after admissible
modification of $\pi^+$ on the post-$t$ trajectory we retain
$X \ge Y + (\delta - \eta) > Y$ on $G$ with positive probability.
Condition~(iii) yields strict $\rho$-monotonicity for this pair:
$\rho_t^{M^\star}(X) > \rho_t^{M^\star}(Y)$. Therefore
$\Delta_t(a^+ \mid h_t) \ge
[\rho_t^{M^\star}(X) - \rho_t^{M^\star}(Y)]_+ > 0$.

\emph{Set-level (b).} Since $\A^+_t = \A^0_t \cup \{a^+\}$, for each
$M \in \M$ we have
$\sup_{a \in \A^+_t}\rho_t^M(L^{\doop(a), \pi^+}) =
\max\bigl(\sup_{a \in \A^0_t}\rho_t^M(L^{\doop(a), \pi^+}),
\rho_t^M(L^{\doop(a^+), \pi^+})\bigr)$.
Using $\sup_M\max(f, g) = \max(\sup_M f, \sup_M g)$,
\[
K_t(\A^+_t \mid h_t)
\;=\;
\max\!\Bigl(K_t(\A^0_t \mid h_t),\;
\sup_{M \in \M}\rho_t^M(L^{\doop(a^+), \pi^+})\Bigr).
\]
Strict greater than $K_t(\A^0_t \mid h_t)$ iff the second argument
strictly exceeds the first, proving (b) as an if-and-only-if.
\end{proof}

\begin{remark}[Why the split matters]\label{rem:iap-split}
Without distinguishing $\Delta_t$ from $K_t$, the naive claim ``adding an
irreversible action raises robust capital'' is false: if $\A^0_t$
already permits an action whose worst-case risk exceeds that of $a^+$,
then $K_t(\A^+_t) = K_t(\A^0_t)$. The defensible content is that
irreversible non-hedgeable authority has a strictly positive incremental
price relative to its own safe default, and additionally raises
contract-level capital only when it becomes binding. This decomposition
mirrors the distinction between marginal and total risk contributions in
classical capital allocation \citep{tasche2007euler}.
\end{remark}

\section{Runtime Risk Gating with Conservative Upper Bounds}\label{sec:gating}

Exact computation of $c_t(a \mid h_t)$ at every action is prohibitive in
deployment. Let $\bar c_t(a \mid h_t)$ be a conservative \emph{nonnegative}
upper envelope on the positive toll $c_t^+(a \mid h_t)$, provided by
simulation, a calibrated quantile model, or distributionally robust
optimisation.

\begin{assumption}[High-probability toll envelope]\label{ass:envelope}
Fix a confidence level $1 - \delta \in (0, 1)$. The envelope
$\bar c_t \ge 0$ on a set of adaptively evaluated query points
$\mathcal{I} \subseteq \{(t, a_t)\}$ satisfies
\[
\Pbb\!\left(\,
c_t^+(a_t \mid h_t) \le \bar c_t(a_t \mid h_t)
\text{ for all } (t, a_t) \in \mathcal{I}
\right)\ge 1 - \delta.
\]
For envelopes constructed by conformal prediction, the calibration scheme
must match the deployment adaptivity assumptions: split conformal
\citep{vovk2005algorithmic,angelopoulos2021gentle} is sufficient under
exchangeable or fixed evaluation, while \emph{online adaptive action
selection}---in which the policy itself depends on $\bar c$---generally
breaks exchangeability and requires an online or adaptive conformal
variant \citep{gibbs2021adaptive}.
\end{assumption}

\begin{theorem}[Runtime risk-budget guarantee]\label{thm:gating}
Let an agent begin with toll budget $B_0 \ge 0$. A runtime gate executes
action $a_t$ at time $t$ only if
\[
\bar c_t(a_t \mid h_t) \le B_t,
\]
in which case the budget is updated $B_{t+1} := B_t - \bar c_t(a_t \mid
h_t)$; otherwise the gate blocks, downgrades to safe default $a^0_t(h_t)$,
or escalates to human approval. Under Assumption~\ref{ass:envelope},
with probability at least $1 - \delta$,
\[
\sum_{(t, a_t) \in \mathcal{I}_{\mathrm{exec}}}
c_t^+(a_t \mid h_t)
\;\le\; B_0,
\]
where $\mathcal{I}_{\mathrm{exec}}$ is the set of gate-admitted executed
actions.
\end{theorem}

\begin{proof}[Proof sketch]
Let $\mathcal{E}_\delta = \{c_t^+(a_t \mid h_t) \le \bar c_t(a_t \mid h_t)
\text{ for all } (t, a_t) \in \mathcal{I}\}$; by
Assumption~\ref{ass:envelope}, $\Pbb(\mathcal{E}_\delta) \ge 1 - \delta$.
Work on $\mathcal{E}_\delta$. The gate executes only when $\bar c_t \le
B_t$ and updates $B_{t+1} = B_t - \bar c_t$. Since $\bar c_t \ge 0$, the
budget recursion is monotone-nonincreasing, $B_t \ge 0$ throughout, and
summing across executed actions gives
\[
\sum_{(t, a_t) \in \mathcal{I}_{\mathrm{exec}}}
\bar c_t(a_t \mid h_t)
\;=\; B_0 - B_{\mathrm{final}}
\;\le\; B_0.
\]
On $\mathcal{E}_\delta$, $c_t^+ \le \bar c_t$ on every evaluated point,
hence on every executed point, so $\sum c_t^+ \le \sum \bar c_t \le B_0$.
All statistical difficulty is isolated in the validity of
Assumption~\ref{ass:envelope} under the deployment adaptivity regime.
\end{proof}

\begin{remark}[Why this paper does not claim policy dominance]\label{rem:no-fsd}
A tempting stronger claim is that adding tolls makes the agent's
terminal-loss distribution first-order stochastically dominate the
untolled baseline. That claim is \emph{false} in a general MDP for at
least three reasons: (a) the reward may be perversely correlated with
tail loss, so risk-toll-induced changes to the policy can shift mass to
worse outcomes elsewhere; (b) the action space may lack any lower-risk
substitute for the priced action under the relevant history, so the
penalised policy collapses to refusal rather than substitution; (c)
local one-step risk-toll penalties need not imply global distributional
dominance even when monotone in expectation. The defensible structural
statement of this framework is the budget guarantee
(Theorem~\ref{thm:gating}). Any policy-improvement result should be
stated either as an empirical Paper~B outcome or under explicit
reward--loss alignment and action-richness assumptions, following the
risk-sensitive MDP literature \citep{chow2014algorithms,tamar2015policy}.
\end{remark}

\section{Companion Layers and Program Interfaces}\label{sec:program}

The preceding sections define the foundational contract objects: the
side-effect-bearing action, the minimal-authority safe default, the
underwriting boundary, the counterfactual toll, the boundary potential, the
irreversible-authority premium, and the conservative gate. The rest of the
research program keeps those primitives fixed and asks three different
questions.

\begin{table}[ht]
\centering
\small
\begin{tabular}{p{0.24\linewidth}p{0.68\linewidth}}
\toprule
Layer & Role in the program \\
\midrule
Empirical runtime &
Implements the toll gate as an Actuarial Action Interface and evaluates
whether operational authority is released at bounded actuarial capital
across database, refund, and tool-agent replay environments. \\
Mechanism design &
Asks whether strategic operators can reduce payable tolls by splitting,
misreporting, or exploiting interface failures; strengthens the
boundary-design corollary into aggregate-settlement, ambiguity-reserve,
and premium-aware incentive-compatibility clauses. \\
Dynamic underwriting &
Asks how premiums and reserves should update after observed experience;
turns the conservative envelope and audit telemetry into an
experience-rated contract with credibility learning and audit-replay
calibration. \\
\bottomrule
\end{tabular}
\caption{Program interfaces. Paper~A supplies the mathematical accounting
layer; companion papers instantiate, strategically harden, and dynamise the
same contract.}
\label{tab:program}
\end{table}

This division prevents the foundational claims from becoming overbroad.
Paper~A does not claim that tolls make every policy safer, that a finite
benchmark exhausts agent risk, or that strategic manipulation is impossible
without additional contract clauses. It claims the narrower accounting
facts proved above: the toll is well-defined once primitives are fixed,
within-boundary splitting telescopes through a boundary potential,
irreversible authority has positive marginal price when the stated witness
conditions hold, and a gate that charges a valid conservative envelope
cannot spend more positive toll than its initial budget on the covered
execution path.

\paragraph{Empirical status.}
The empirical companion \citep{chen2026insuring} has moved beyond the
original deterministic database prototype. It instantiates the same safe-default and budget
recursion through a deterministic runtime interface, paired proposal-replay
experiments, Postgres/LangChain stress tests, four-environment
authority-frontier summaries, and a live LLM underwriting panel. Those
results are intentionally not re-reported here; their role is to test the
operational relevance of the mathematical objects introduced in this
paper.

\section{Residual Obligations and Scope}\label{sec:obligations}

The proof sketches above isolate the assumptions under which the four
foundational results hold. The following obligations remain outside the
scope of this paper, or are handled by companion layers rather than by
changing the present theorems:
\begin{enumerate}[leftmargin=1.5em,itemsep=2pt]
  \item \textbf{Model-class verification.} Strict
  $\rho$-monotonicity for the witness pair under $M^\star$
  (Definition~\ref{def:irr}(iii)) must be verified for any chosen runtime
  risk family. The entropic family and conditional $\ES_\alpha$ require
  separate verification of their strictness regimes.
  \item \textbf{Closed-loop calibration.} Assumption~\ref{ass:envelope}
  is a high-probability envelope assumption. Split conformal methods
  supply it under exchangeability; closed-loop deployment requires an
  online-valid or audit-replay calibration protocol
  \citep{gibbs2021adaptive}. The dynamic-underwriting companion develops
  one audit-replay route, but this paper does not prove a general
  adversarial conformal theorem.
  \item \textbf{Boundary failure and cross-boundary settlement.}
  Theorem~\ref{thm:no-split} is strictly within-boundary. Cross-session,
  related-party, proxy-agent, and multi-boundary arbitrage require
  explicit aggregate settlement clauses. These are mechanism-design
  questions built on top of Corollary~\ref{cor:gaming}, not hidden
  consequences of Theorem~\ref{thm:no-split}.
  \item \textbf{Policy welfare.} Theorem~\ref{thm:gating} is a reserve
  accounting guarantee, not a welfare theorem and not a first-order
  stochastic dominance statement. Any claim that tolls improve outcomes
  requires empirical evidence or additional reward--loss alignment and
  action-richness assumptions.
\end{enumerate}

\begingroup
\raggedright
\bibliographystyle{plainnat}
\bibliography{references}
\endgroup

\end{document}